\documentclass{llncs}
\usepackage{amsfonts}
\usepackage{amssymb}
\usepackage{amsmath,bm}
\usepackage{graphicx}
\usepackage{algorithm}
\usepackage{algorithmic}
\usepackage{multirow}
\usepackage{amsmath}
\usepackage{xcolor}
\usepackage{subfigure}

\begin{document}
%

\title{An Evolutionary Approach for Optimizing Hierarchical Multi-Agent System Organization
}

%
%
%
%
%

%
\author{
%
%
Zhiqi Shen, Ling Yu, Han Yu\\
\institute{Nanyang Technological University, Singapore}
}

\maketitle
\begin{abstract}
It has been widely recognized that the performance of a multi-agent system is highly affected by its organization. A large scale system may have billions of possible ways of organization, which makes it impractical to find an optimal choice of organization using exhaustive search methods. In this paper, we propose a genetic algorithm aided optimization scheme for designing hierarchical structures of multi-agent systems. We introduce a novel algorithm, called the hierarchical genetic algorithm, in which hierarchical crossover with a repair strategy and mutation of small perturbation are used. The phenotypic hierarchical structure space is translated to the genome-like array representation space, which makes the algorithm genetic-operator-literate. A case study with 10 scenarios of a hierarchical information retrieval model is provided. Our experiments have shown that competitive baseline structures which lead to the optimal organization in terms of utility can be found by the proposed algorithm during the evolutionary search. Compared with the traditional genetic operators, the newly introduced operators produced better organizations of higher utility more consistently in a variety of test cases. The proposed algorithm extends of the search processes of the state-of-the-art multi-agent organization design methodologies, and is more computationally efficient in a large search space.
\end{abstract}




\section{Introduction}
The organization of a multi-agent system (MAS) provides a framework for agent activities and interactions through the definition of agent roles, groups, tasks, behavioral expectations and authority relationships. A proper organization can ensure the behavior of the agents to be externally observable \cite{Ferber-et-al:2004}. Particularly, in large scale systems such as crowdsourcing systems \cite{Yu-et-al:2012}, to form and evolve an organization makes it possible for the system to exploit collective efficiencies and to manage emerging situations \cite{Lesser:1998}. Experiments and simulations have shown that various organizations employed by a system with the same set of agents may have different impacts on its performance \cite{Okamoto-et-al:2008,Fernandez-Ossowski:2008,Horling-Lesser:2008,Sims-et-al:2008,Zafar-et-al:2008,Zhang-et-al:2013,Yu-et-al:2013IEEE,Lesser-Corkill:2014}.

Among all kinds of organizations, the hierarchical structure is one of the most common structures observed in multi-agent systems. Many multi-agent systems can be abstracted as hierarchical, tree-like structures or sets of parallel hierarchical structures, where agents are categorized in different levels in the hierarchies \cite{Kirley:2006}. Often, the level of an agent indicates its capabilities and roles. A specific level in the system consists of equally capable agents, performing similar roles, as seen in the distributed information retrieval (IR) system described in \cite{Horling-Lesser:2008}.

For a large hierarchical MAS, there exist a great variety of possible ways to organize the system. Due to the difference in the depth and the width of the hierarchy, the number of organization instances increases exponentially with the number of agents. Although many methodologies for organization modeling have been proposed, few of them present an effective way to search for an optimal organization instance. In order to solve the problem, this paper proposes a genetic algorithm (GA) approach as an alternative to the conventional enumeration methods for optimizing hierarchical multi-agent systems. Inspired by biological evolution processes such as selection, reproduction, and mutation, GAs are known to be robust global search algorithms for optimization and machine learning \cite{Holland:1992,Back:1996,Mezura-Montes:2007}. The heuristic nature of GA helps it to locate the global optimum in a vast search space. We design novel crossover and mutation operators to make the algorithm suitable for organization evolution and thereby ensure competitive performance. We will test the algorithm in an example of the IR model \cite{Horling-Lesser:2008} which exhibits numerous possible organizational variants and verify its capability through simulations in different scenarios.

\section{Related Work}
The design of a multi-agent system organization has been investigated by many researchers. Early methodologies such as Gaia \cite{Wooldridge-et-al:2000} and OMNI \cite{Vazquez-Salceda:2005} aim to assist the manual design process of agent organizations. Instead of relying heavily on the expertise of human designers, it is desirable to automate the process of producing multi-agent organization designs. In this sense, a quantitative measurement of a set of metrics is needed to rapidly and precisely predict the performance of the MAS. With these metrics we can evaluate a number of organization instances, rank them, and select the best one without introducing heavy cost by actually implementing the organization designs.

In \cite{Horling-Lesser:2008}, an organizational design modeling language (ODML) was proposed, and the utility value was defined as the quantitative measurement of the performance of a distributed sensor network and an information retrieval system. Several approaches, including the exploitation of hard constraints and equivalence classes, parallel search, and the use of abstraction, have been studied in order to reduce the complexity of searching for a valid optimal organization.

Another organization designer, KB-ORG, which also incorporates quantitative utility as a user evaluation criterion, was proposed for multi-agent systems in \cite{Sims-et-al:2008}. It uses both application-level and coordination-level organization design knowledge to explore the search space of candidate organizations selectively. This approach significantly reduces the exploration effort required to produce effective designs as compared to modeling and evaluation-based approaches that do not incorporate designer expertise.

Nonetheless, similar to ODML, KB-ORG aims at pruning the search space. The design knowledge alone is inadequate for the identification of an optimal design when the possible variety of the organization structure becomes large.

Evolutionary based search mechanisms have been used to help the design of MAS organizations on a few occasions. For example, in \cite{Yang-Luo:2007}, a GA-based algorithm is proposed for coalition structure formation which aims at achieving the goals of high performance, scalability, and fast convergence rate simultaneously. And in \cite{Li-et-al:2009}, a heuristic search method, called evolutionary organizational search (EOS), which is based on genetic programming (GP), was introduced. A review of evolutionary methodologies, mostly involving co-evolution, for the engineering of multi-agent market mechanisms, can also be found in \cite{Phelps-et-al:2010}. These techniques show a promising direction to deal with the organization search in hierarchical multi-agent systems, as exhaustive methods, such as breadth-first search and depth-first search, become inefficient and impractical in a large search space.

\section{Organization Representation}
Generally speaking, the organization of a hierarchical MAS consists of a number of tree structures. According to the number of "leaders", it can be either a single tree or a set of trees. The intermediate nodes in a tree are responsible of assigning tasks to their subordinates, as well as reporting back to their higher-level authorities. Information exchange is only allowed in the vertical directions between higher and lower levels. There is no interaction of agents horizontally, or among different hierarchies. The leaf nodes are the bottom of the structure and they complete the most basic tasks.

Optimization in such a search space can be handled by evolutionary algorithms \cite{Mezura-Montes:2007}, especially genetic programming, which supports populations of tree structures. It has also been shown that some well-structured trees (e.g. binary trees), with a certain number of levels and a fixed number of subordinates per node, can be represented by arrays \cite{Nan-et-al:2005,Aranha-Iba:2009}. Transformations are feasible as a result of their regular structures, which allow the traditional crossover and mutation operators of other evolutionary algorithms, such as genetic algorithms, to take effect.

We propose an array representation of hierarchical MAS organizations which is applicable to a much broader range of hierarchical structures than just binary trees. It converts s set of hierarchical trees into a fixed-length array with integer components, which resemble gene sequences. The representation is not limited to describe a single tree, and the number of subordinates of each node need not be a constant. Unbalanced trees, in which leaf nodes are not on the same hierarchical level, can also be depicted using this representation.

\subsection{Translating Organizations into Genomes}
We assume that the hierarchical MAS considered here have the following properties. We assume that the number of leaf node agents is fixed before the search. We also assume that the maximum possible number of levels is determined. Thus, the total number of agents in the organization is bounded. Based on these assumptions, we can make use of the partition concept to convert the organization from tree structures to arrays.

Let N be the total number of leaf nodes or end nodes, so that the they can be numbered as $1,2,...,N$ respectively from left to right. Let M be the maximum tree depth (i.e. maximum height of the structure). The reason for limiting the height is that very tall structures can be slow or irresponsive, as the long path length from root to leaf increases message latency among the agents. The organization of a hierarchical MAS can be outlined by Representation 1:
\begin{equation}
a_1a_2a_3...a_{N-1}
\end{equation}
where $a_i$ is an integer between 1 and $M$, denoting the level number where leaf nodes $i$ and $i+1$ start to separate.

An example with seven leaf nodes ($N=7$) is illustrated in Figure \ref{fig:1}. It consists of two trees. On Level 1, the four leaf nodes on the left and the three leaf nodes on the right separate into two trees. In other words, there is a separation between the leaf nodes 4 and 5, so $a_4=1$. On Level 2, there are two leaf nodes and one intermediate node (three nodes altogether) under the left tree root, corresponding to the ``2 2" (two partition numbers) to the left of the ``1" in the array. The one leaf node and one intermediate node (two nodes altogether) under the right tree root give the ``2" (one partition number) to the right. Both intermediate nodes on Level 2 have two leaf nodes as their subordinates (leaf nodes 3 and 4, leaf nodes 6 and 7), which are separated on Level 3, resulting in the two 3's in the 3rd and 6th places in the array. Therefore, we get the whole array ``2 2 3 1 2 3" for the organization.

Conversely, we can also obtain an organization by interpreting the representation array. For instance, if we want to determine which level node 4 in Figure \ref{fig:1} sits on, we need to examine both the node's left and right neighbor. The 3rd and 4th digits in the array are ``3" and ``1". It means that node 3 and node 4 are separated on Level 3. Node 4 and node 5 are separated on Level 1. As a result, node 4 is on Level 3 (larger number between 3 and 1). Similarly, because the fifth digit is ``2", i.e. node 5 and node 6 are separated on level 2, node 5 should be on level 2 (larger number between 2 and 1).

\begin{figure}[t!]
    \includegraphics[trim = 40mm 240mm 40mm 25mm, clip, width = 5in]{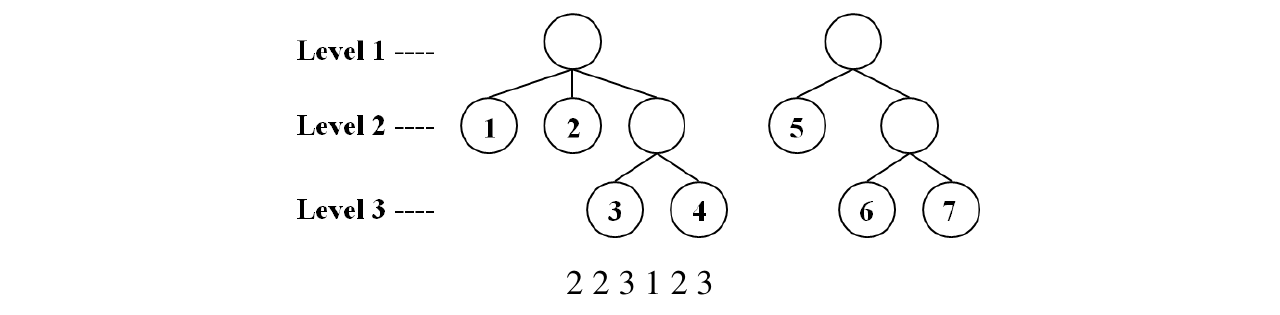}
    \caption{A sample organization and its array representation. Agent nodes are displayed as circles in the figure. Leaf nodes are numbered.}\label{fig:1}
\end{figure}

\begin{theorem}
The above representation has the following properties:
\begin{enumerate}
    \item For every hierarchical organization instance which satisfies our assumptions in the beginning of Section 3.1, the array representation that can be generated is unique.
    \item For every representation of the above mentioned form, there is an organization instance corresponding to it.
\end{enumerate}
\label{Them:TR1}
\end{theorem}

\begin{proof}

\begin{enumerate}
    \item We firstly prove the existence of an array representation for every hierarchical organization instance. The way of generating an array representation of an arbitrary hierarchical organization instance can be expressed as follows. If there are $N$ leaf nodes, we prepare $N–1$ slots. Firstly, organize the structure well so that the root nodes, intermediate nodes, and leaf nodes are on their proper levels. Secondly, we examine the separation pattern between adjacent leaf nodes one by one from left to right. Fill the slots with the level number where the adjacent leaf nodes start to separate. See Figure \ref{fig:1} for an example. The first two leaf nodes on the left are direct subordinates of the first tree root, i.e. on the root level (Level 1) they do not separate. However, on Level 2, they separate into different nodes. So the first number is 2. The second slot should also be filled with 2 because the second and third leaf nodes on the left separate on Level 2. And as the third and fourth leaf nodes are direct subordinates of an intermediate node on Level 2, they start to separate on Level 3. 3 should be the third number in the array. And so on, we can get the values, which are the level numbers, for all the slots. Together they form the required representation.

        We then prove the uniqueness of the generated array representation. If array representations $a_1a_2a_3...a_{N–1}$ and $b_1b_2b_3...b_{N–1}$ which are derived from the same organization instance are different, there exits an $i\in{1,2,...,N}$ such that $a_i\neq b_i$. This shows that the leaf nodes $i$ and $i+1$ separate at different levels in the two corresponding organization structures, which means the organization structures are not identical.
    \item Given an array representation with positive integers of length $L$, we would like to construct an organization instance containing $L+1$ leaf nodes as follows. Find all the digit ``1"s in the representation (if there are any). Calculate the number of digits (greater than 1) between adjacent 1's one by one from left to right, and denote them as $n_1,n_2,n_3,...,n_{k+1}$, where $k$ is the number of 1's. If there are no 1's, then $k=0$ and $n_1=L$. The corresponding organization has $k+1$ root nodes with $n_1+1,n_2+1,n_3+1,...,n_{k+1}+1$ leaf nodes, respectively, from left to right. So far we have completed the root level (Level 1) of the organization. For instance, with array [2 2 3 1 2 3], $n_1=3$, $n_2=2$, i.e. there are two root nodes with 4 and 3 leaf nodes respectively. For Level 2, we take segments with 1's and 2's as separators. These segments should only contain digits greater than 2 (if any). Like what is done for Level 1, the number of digits between adjacent separators are recorded as $r_1,r_2,r_3,...,r_{t+1}$, where t is the total number of 1's and 2's. If $r_i=0$, it corresponds to a leaf node; otherwise, it corresponds to an intermediate node on Level 2. After that, take segments with 1's, 2's, and 3's as separators, and repeat the steps until the greatest numbers in the representation are examined. In this way we can obtain the full organization instance.
\end{enumerate}

\end{proof}

\begin{figure}
    \includegraphics[trim = 40mm 240mm 40mm 25mm, clip, width = 5in]{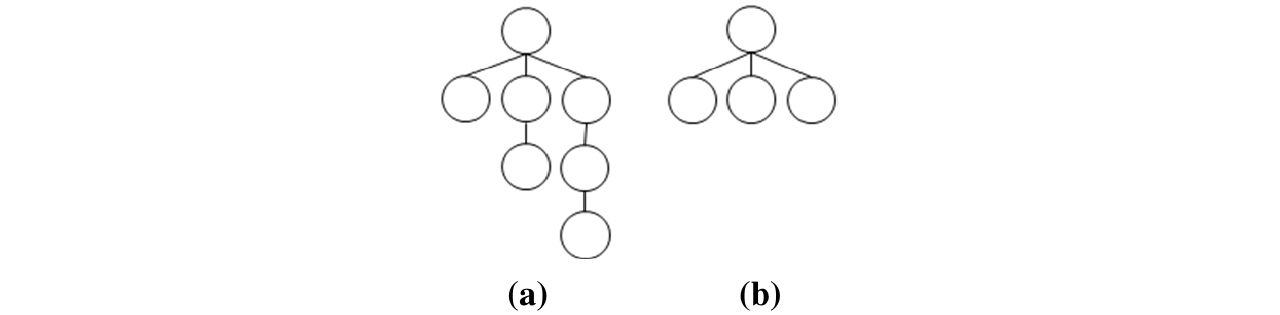}
    \caption{Organizations with the same representation.}\label{fig:2}
\end{figure}

Note that the organization instance is non-unique. Figure \ref{fig:2}(a) illustrates an extreme case where all three leaf nodes separate on Level 2, so the representation is [2 2]. It has the same representation as the organization in Figure \ref{fig:2}(b). When such circumstances arise, we should examine all the possible organization instances that correspond to a representation and use the best one. In the following sections we explain that in the IR model, the sub-organizations having nodes with only one subordinate are uneconomical and should be simplified to achieve higher utility.

So far, we have established a surjective mapping from the set of all valid structure instances containing $N$ leaf nodes with maximum height $M$, denoted as $A$, to the set of all arrays containing $N–1$ integer elements ranging from 1 to $M$, denoted as $B$. Furthermore, the representation is compatible with genetic operators such as one-point, two-point or uniform crossover, i.e. the offspring generated after the crossover of individuals from set $B$ still belong to set $B$. Bit-wise mutation can also be applied here, so that every bit of the genome ai is mutated to a randomly picked different value from $\{1,2,...,M\}/\{a_i\}$ according to the user defined mutation probability.

\subsection{Simplifying Organizations}
The above representation can be applied to a general hierarchical MAS organization. For specific organization search problems, we may find it beneficial to simplify the representation in order to prune the search space and avoid unnecessary candidate evaluations of the algorithm. Trimming, combining, and reducing of branches are easy to achieve using the proposed representation. We will give an example of how to remove redundant intermediate nodes of the IR system in Section 5.1.

\subsection{Variations of Representations}
In Section 3.1, we have assumed that the leaf nodes are homogeneous. In such circumstances, an $N-1$ array is enough to represent a hierarchical organization of a MAS. Nonetheless, in view of the circumstances where each leaf node must be treated uniquely, a second row can be added to the array representation to address the distinction resulting from permutations. This will make the representation to be in the form of a $2(N-1)$ array (Representation 2):
\begin{align*}
\left( \begin{array}{c}
a_1a_2a_3...a_{N-1} \\
p_1p_2p_3...p_{N-1}
\end{array} \right)
\end{align*}
where $\{a_i\}$ are still integers between 1 and $M$, denoting the level of the partition between leaf nodes $i$ and $i+1$, and $p_1,p_2,...,p_{N–1}$ are a permutation of 1 to $N$ with the last number discarded. Still using the example in Figure \ref{fig:1}, now we use numbers 1, 2, ..., 7 to distinguish the mutually different leaf nodes. If in the organization they are 5, 3, 2, 1, 4, 7, 6, respectively, then the representation is:
\begin{align*}
\left( \begin{array}{cccccc}
2 & 2 & 3 & 1 & 2 & 3\\
5 & 3 & 2 & 1 & 4 & 7
\end{array} \right)
\end{align*}
One may also want to design an organization in which the number of leaf node agents is not fixed beforehand. To account for varied number of leaf node agents, we may use the following Representation 3:
\begin{align*}
a_1a_2a_3...a_{N_1-1}\underbrace{00...0}_{(N_2-N_1)}
\end{align*}
where $N_1$ is the actual number of leaf nodes of the representation, $N_2$ is the maximum number of leaf nodes allowed in the organization, and the remaining positions are filled with zeros. These variants of representations will function in the same manner as the Representation 1 when taken to go through genetic operators which are introduced next.

\section{Crossover and Mutation Operators}

\begin{figure*}
    \subfigure[Array representation.]{
    \includegraphics[trim = 35mm 235mm 35mm 25mm, clip, width = 5in]{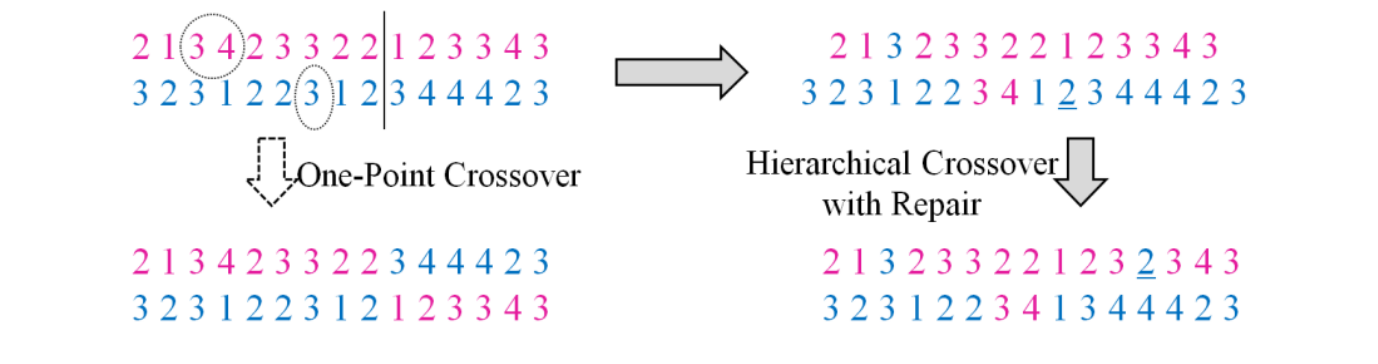}}

    \subfigure[One-point crossover.]{
    \includegraphics[trim = 39mm 229mm 39mm 25mm, clip, width = 5in]{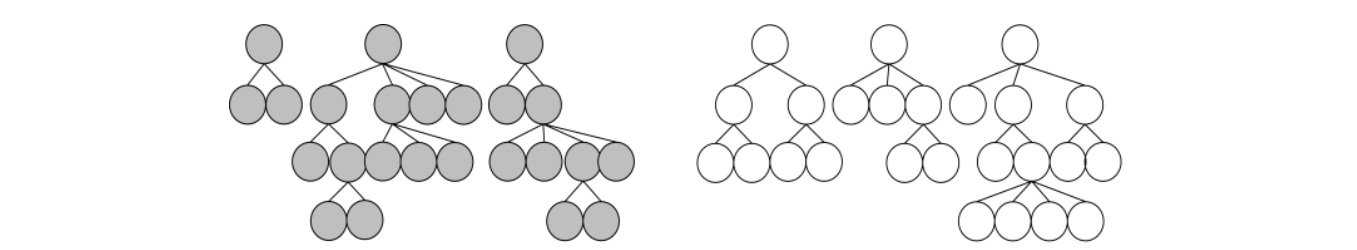}}

    \subfigure[Hierarchical crossover.]{
    \includegraphics[trim = 30mm 149mm 30mm 25mm, clip, width = 5in]{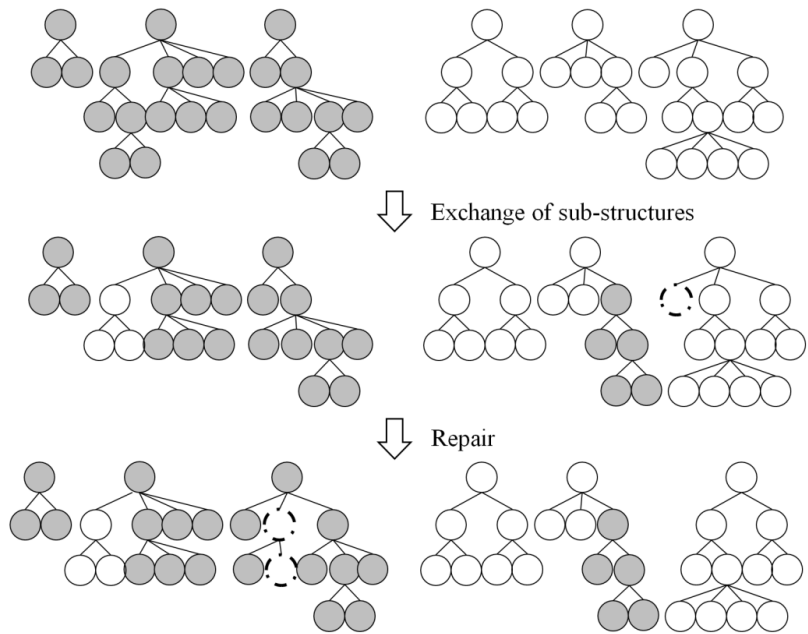}}

    \caption{Illustration of one-point crossover and hierarchical crossover using array representation and organization structures.}\label{fig:3}
\end{figure*}

The traditional one-point crossover chooses a random slicing position along the chromosomes of both parents. All data beyond that point in either solution is swapped between the two parents. The resulting chromosomes are two offspring. Though commonly used in genetic algorithms, this crossover method only influences the structure near the crossover point, as shown in Figure \ref{fig:3}(a,b).

It may not be enough to generate new offspring in large-scale systems. To speed up the evolution and increase the chance of getting a desired structure with higher utility, new crossover operators are needed. In this paper, we propose a novel crossover operator, hierarchical crossover, specially designed for optimization of tree-structured organizations. The proposed hierarchical crossover operator based on the previously described Representation 1 contains swapping of sub-organizations and a repair strategy to keep the number of total leaf nodes constant. It is implemented as follows.

First of all, we compare the number of structure levels of two randomly selected organization solutions from the population. Denote the organization with more levels as the first individual and the number of levels as $T$. Denote the organization with fewer levels as the second individual. (In the case of a tie, the order can be arbitrarily assigned.) After that, we choose a node randomly from all nodes whose level number is between 1 and $T-1$ from the first solution and denote the level number of the chosen node as $S$. Thirdly, we choose a node randomly at Level $S$, or the penultimate level, whichever is smaller, from the second solution, and exchange the sub-structures between the two solutions below the chosen nodes. If any of the solution candidates have only one level, we generate two random individuals of maximum tree depth instead. The exchange ensures that the two newly formed organization structures do not exceed the maximum height of their parent structures. However, the exchanged sub-structures do not necessarily contain equal number of leaf nodes. Thus, we propose the following repair strategy.

\begin{algorithm}[!t]
\caption{Hierarchical crossover}
\begin{algorithmic}[1]\label{al:1}
\STATE Let $p_1$ and $p_2$ be the array representations of two selected parents.
\IF {$\max(p_1)<\max(p_2)$}
    \STATE Exchange $p_1$ and $p_2$;
\ENDIF
\STATE $T=\max(p_1)$;
\IF {$T==1$ or $\max(p_2)==1$}
    \STATE Randomly generate two offsprings, $o_1$ and $o_2$, of maximum tree depth;
\ENDIF
\STATE List all possible crossover nodes of $p_1$ from Level 1 till $T-1$;
\STATE Randomly select a node from the above list as $cp_1$;
\STATE Record the level number of $cp_1$ as $S$;
\STATE Get the segments of the array representation of the sub-structure below $cp_1$ as $ss\_c1$;
\STATE Get the segments of the array representation to the left of the sub-structure below $cp_1$ as $ss\_l1$;
\STATE Get the segments of the array representation to the right of the sub-structure below $cp_1$ as $ss\_r1$;
\STATE Randomly select a node $cp_2$ from $p_2$ at the Level No. $\min[S, \max(p2)-1)]$;
\STATE Get the segments of the array representation of the sub-structure below $cp_2$ as $ss\_c2$;
\STATE Get the segments of the array representation to the left of the sub-structure below $cp_2$ as $ss\_l2$;
\STATE Get the segments of the array representation to the right of the sub-structure below $cp_2$ as $ss\_r2$;
\STATE $o_1=[\begin{array}{ccc} ss\_l1 & ss\_c2 & ss\_r1 \end{array}]$;
\STATE $o_2=[\begin{array}{ccc} ss\_l2 & ss\_c1 & ss\_r2 \end{array}]$;
\IF {length($o_1$) $>$ length($p_1$)}
    \STATE exnum = length($o_1$) - length($p_1$);
    \FOR {$j=1:$ exnum}
        \STATE Randomly select an integer $k_1$ between 1 and length($o_1$);
        \STATE Randomly select an integer $k_2$ between 1 and length($o_2$)+1;
        \STATE $o_2=[\begin{array}{ccc} o_2(1:k_2-1) & o_1(k_1) & o_2(k_2:end) \end{array}]$;
        \STATE $o_1=[\begin{array}{cc} o_1(1:k_1-1) & o_1(k_1+1:end) \end{array}]$;
    \ENDFOR
\ELSIF {length($o_2$) $>$ length($p_2$)}
    \STATE exnum = length($o_2$) - length($p_2$);
    \FOR {$j=1:$ exnum}
        \STATE Randomly select an integer $k_2$ between 1 and length($o_2$);
        \STATE Randomly select an integer $k_1$ between 1 and length($o_1$)+1;
        \STATE $o_1=[\begin{array}{ccc} o_1(1:k_1-1) & o_2(k_2) & o_1(k_1:end) \end{array}]$;
        \STATE $o_2=[\begin{array}{cc} o_2(1:k_2-1) & o_2(k_2+1:end) \end{array}]$;
    \ENDFOR
\ENDIF
\end{algorithmic}
\end{algorithm}

Find the solution with longer representation and randomly pick out one digit from it and insert this digit into a random slot in the other solution. Continue until the two solutions have equal length. This will guarantee the validity of the two solutions, as shown in Figure \ref{fig:3}(a,c). Illustrated in both the array representation and the organization structures, Figure \ref{fig:3} displays the difference between the proposed hierarchical crossover and one-point crossover. The pseudo code of hierarchical crossover is given in Algorithm \ref{al:1}.

To apply hierarchical crossover to Representation 2, all we need is to bundle each column and move the second row together with the first row. As for organizations in Representation 3, the repair strategy is implemented with the digits randomly picked out from non-zero locations only and until each selected organizations have the same number of leaf nodes as before.

As seen from Figure \ref{fig:3}, a branch of the tree is corresponding to a piece of gene fragment. By swapping the two selected gene segments in the parents, we get two new organization instances with exchanged sub-organizations. This step is similar to two-point crossover, in which the segments between the two randomly selected crossover points of both parents are swapped to form the offspring. However, like one-point crossover, two-point crossover also does not concern whether the selected gene segments correspond to the whole tree branches or not. And as long as the two crossover points are determined, the locations of the segments in the arrays do not change. Hierarchical crossover is different from two-point crossover in that it focuses on the branches of the tree structures and only change the gene segments that refer to whole branches. Moreover, the locations of the two gene segments of the parents may differ from each other, and the repair strategy promotes organization update.

In addition to the crossover method mentioned above, we use the mutation of small perturbation. It is different from bit-wise mutation in that the digit can only increase by 1 or decrease by 1 with equal probability. In the cases of the boundaries, if the perturbed digit is out of bounds, the original value is restored. The pseudo code of the mutation operator based on Representation 1 is displayed in Algorithm \ref{al:2}.

\begin{algorithm}[!t]
\caption{Mutation of small perturbation}
\begin{algorithmic}[1]\label{al:2}
\STATE Let $os$ be the array representation of an offspring created by the crossover operator, $numVar$ be the length of the representation, $mutOps$ be the mutation probability, and $maxTreeDepth$ be the maximum tree depth.
\STATE $rN=rand(size(os,1),numVar)<mutOps$;
\STATE $os=os+rN\times((rand(size(os,1),numVar)>0.5)\times2-1)$;
\STATE $os(os==0)=1$;
\STATE $os(os==maxTreeDepth+1)=maxTreeDepth$;
\end{algorithmic}
\end{algorithm}

\section{The Information Retrieval Model}
In this paper we will examine the algorithm in the information retrieval system \cite{Horling-Lesser:2008}. A structured, hierarchical organization composed of nodes as mediators, aggregators, and databases is used to model the IR system. An agent is assigned for each node to take the corresponding functions. The information recall and the query response time are combined to form a metric to determine the utility of the organization. Detailed procedures to calculate the utility can be found in \cite{Horling-Lesser:2008}.

At the top level of each hierarchy is a mediator. The user sends a query, which a randomly assigned mediator is responsible to handle. It uses the collection signatures of all the mediators to compare data sources, then routes the query to those mediators that seem appropriate. After the query has been directed through the aggregators and processed by all the databases under the selected mediators, the responsible mediator finally collects and delivers the resulting data.

\subsection{Simplifying Organization Representation with the IR Model}

\begin{figure}[b!]
    \includegraphics[trim = 40mm 189mm 40mm 30mm, clip, width = 5in]{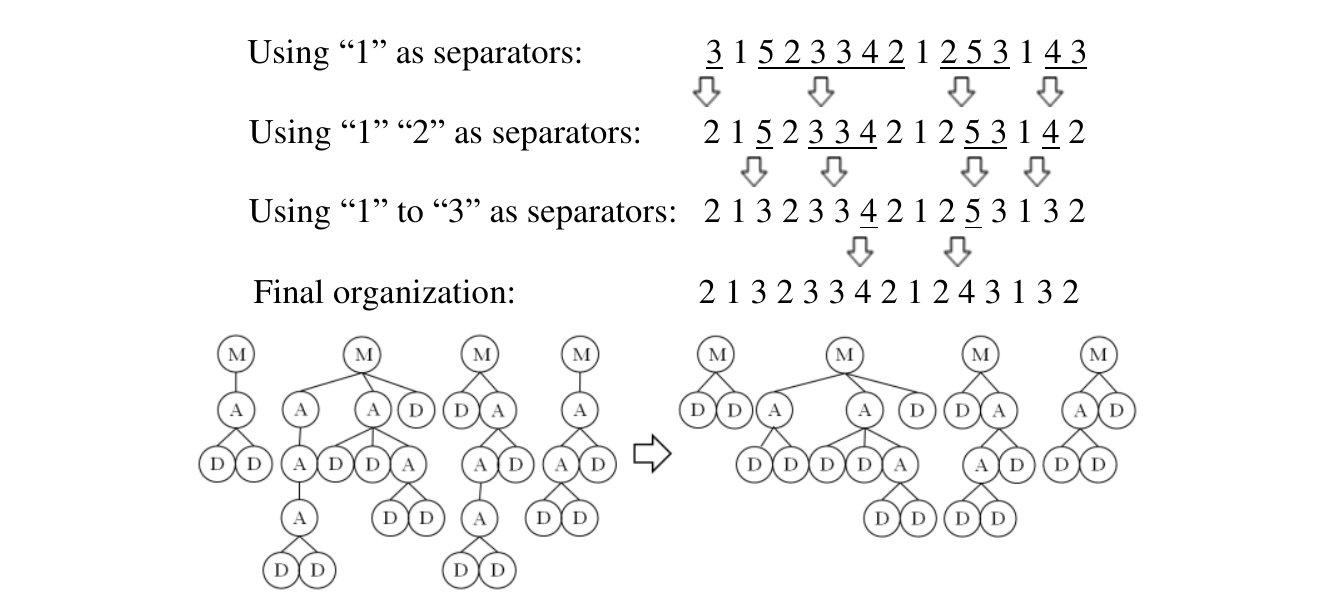}
    \caption{Simplifying the organization. Nodes $M$ are mediators, nodes $A$ are aggregators, and nodes $D$ are databases.}\label{fig:4}
\end{figure}

Since it is assumed in the IR model that all the databases in the system contain the same amount of topic data, and thus, there are no differences among the end nodes (i.e. leaves of the trees), we may apply Representation 1 to the IR model. Here Level 1 is the mediator level, where nodes are all mediators. The intermediate nodes correspond to aggregators, and the leaf nodes are database agents. The whole organization can be outlined by a set of trees.

From a practical viewpoint, we notice that it is not necessary to include an aggregator if it only has one subordinate, because it will only increase the information transmission delay and not bring any integration advantages. Hence, if such an organization instance emerges, we can simply omit the aggregator node and reduce the organization structure by one level.

Related modification can be made in the array representation, which is summarized below. Firstly, obtain all the segments of a genome between adjacent mediators (i.e. the integer series between 1's). Set the smallest values of these segments to 2. Secondly, obtain all the segments with 1's and 2's as separators. Set the smallest values of these segments to 3. Continue until the highest level of the organization. Figure \ref{fig:4} shows the detailed steps of a sample simplifying procedure. It transforms a 5-level sample organization of the IR system to a 4-level one. The simplifying procedure is employed to achieve higher utility. At the same time, the number of organization instances we have to evaluate for every representation is reduced to one.

\subsection{Implementation and Evaluation Criteria}

\begin{table}[!b]
    \caption{Configurations of HGA}\label{tb:1}
    \centering
    \begin{tabular}{|c|c|c|}
    \hline
    No. DBs & Population Size & No. of Candidate Evaluations \\ [0.5ex]
    \hline
    12  &   50      &   2,000 \\
    14  &   100     &   5,000 \\
    16	&   200	    &   10,000 \\
    18	&   500	    &   50,000 \\
    20	&   500	    &   50,000 \\
    22	&   500     &	50,000 \\
    24	&   500     &	100,000 \\
    26	&   500	    &   100,000 \\
    28	&   500	    &   100,000 \\
    30	&   1,000	&   200,000 \\
    \hline
    \end{tabular}
\end{table}

The optimization is carried out using genetic algorithm with population of organizations represented by arrays, the hierarchical crossover and the mutation of small perturbation as described in the above sections. The utility value serves as the fitness measure of an individual organization.

We recognize that there are likely multiple optimal solutions that achieve the same utility in a given system environment, owing to the symmetry of the structures. Therefore, we need a method that allows growth in several promising areas in the search space. In other words, the diversity of the population should be enhanced and over-convergence should be avoided. We increase the competition between similar individuals by applying the restricted tournament selection (RTS) method described in \cite{Harik:1995}. It helps to preserve diverse building blocks needed to locate the optimal organization.

We compare the proposed algorithm, called hierarchical genetic algorithm (HGA), with the standard genetic algorithm using one-point crossover with bit-wise mutation (SGA1) and two-point crossover with bit-wise mutation (SGA2) in order to show the benefits of the newly introduced operators. We examine the algorithms in two aspects, the accuracy and the stability of search, which are evaluated by average percentage relative error (APRE) and success rate (SR) respectively. The percentage relative error (PRE) can be calculated by:
\begin{equation}
PRE=\frac{f_{best}-f}{f_{best}}\times 100
\end{equation}
where $f_{best}$ is the best known fitness value among all the runs of all the algorithms for a given test case, and $f$ is the current fitness value achieved by the algorithm. APRE is the average of the PRE values among all the independent runs of each test case. $SR\in[0,1]$ denotes the ratio of the number of runs in which the best known solution is found by the algorithm to the total number of runs in each test case.

We examine the test cases of 12, 14, 16, 18, 20, 22, 24, 26, 28, and 30 databases. The maximum height of the structures is set to be 4. The population size and the maximum number of candidate evaluations used are shown in Table \ref{tb:1}. All algorithms use a window size $w=5$ for RTS in the population updating stage. The mutation rate is 0.1. All test cases involve 10 independent runs.

The environment parameters are as follows: message latency = 20 milliseconds, process service rate = 10 per second, response service rate = 20 per second, and query rate = 3 per second. The search set size and query set size are set to be the total number of mediators for each organization. The response recall is therefore identical (100\%) in all cases, and the utility is determined by the response time. The computation time of the genetic operators and population updating is negligible compared to that of the candidate evaluations. Therefore, we conclude that the number of candidate evaluations is more suitable as an evaluation. All algorithms are tested in MATLAB 7.9.0.

\section{Experimental Results}
\subsection{Comparison of Results}

\begin{table}[!b]
\caption{Experiment Results (APRE and SR)} \label{tb:results}
\centering
    \begin{tabular}{|c|c|c|c|c|c|c|}\hline
    \multirow{2}{*}{No. DBs} & \multicolumn{2}{|c|}{SGA1} & \multicolumn{2}{c|}{SGA2} & \multicolumn{2}{c|}{HGA} \\ [0.5ex]
                             & APRE & SR & APRE & SR & APRE & SR \\\hline
                        12 & 0.1103 & 0.5 & 0.1122 & 0.5 & {\bf 0.0370} & {\bf 0.8} \\
                        14 & 0.0090	& 0.8 & 0.0460 & 0.7 & {\bf 0} & {\bf 1} \\
                        16 & 0.0966	& 0.7 & 0.0869 & 0.8 & {\bf 0} & {\bf 1} \\
                        18 & 0.0940	& {\bf 0.8} & {\bf 0.0372} & {\bf 0.8} & 0.0505 & {\bf 0.8} \\
                        20 & 0.1150	& {\bf 0.5} & 0.3076 & 0.1 & {\bf 0.0749} & 0.3 \\
                        22 & 0.2037	& 0.1 & 0.3085 & 0 & {\bf 0.0031} & {\bf 0.9} \\
                        24 & 0.3376	& 0.2 & 0.4914 & 0 & {\bf 0.0406} & {\bf 0.9} \\
                        26 & 0.1556	& 0.4 & 0.3494 & 0.1 & {\bf 0} & {\bf 1} \\
                        28 & 0.2104	& 0.2 & 0.5307 & 0 & {\bf 0.0067} & {\bf 0.9} \\
                        30 & 0.2470	& 0.2 & 0.4825 & 0.1 & {\bf 0} & {\bf 1} \\
    \hline
    \end{tabular}
\end{table}

In this section we will demonstrate the advantage of the proposed HGA over the standard GA with one-point and two point crossover in locating the best organization of the IR system. Table \ref{tb:results} shows the APRE and SR values of SGA1, SGA2, and HGA in the 10 test cases. The best value for each test case is highlighted. It can be observed that the accuracy of the proposed HGA is better than SGA1 and SGA2 in 9 out of the 10 cases. Only in the 18-database case, SGA2 outperforms SGA1 and HGA in terms of APRE.

Regarding the search ability, HGA also has an advantage over SGA1 and SGA2 in the majority of the test cases. The superiority of HGA is more pronounced in larger-scale organizations which contain more than 20 database nodes. In those cases, SGA1 and SGA2 fail to locate the best known organization instances for most of the time, whereas the proposed HGA still maintains high SR values of 90\%–100\%. This proves that HGA uses fewer candidate evaluations to locate the best organization than the conventional GAs. Given that the candidate evaluations are very computationally expensive in many real-world systems, it is beneficial to use HGA in such circumstances.

The non-parametric Wilcoxon signed-rank test is performed to judge whether there is a statistically significant difference between HGA and SGA1/SGA2. As a pair-wise test in a multi-problem scenario, we use all the APRE values of each algorithm as sample vectors. The null hypothesis H0 is set as ``there is no difference between HGA and SGA1/SGA2 in terms of the APRE values." Accordingly, the alternative hypothesis H1 is ``The two methods are significantly different." A significance level of 5\% is implemented. We get that the APRE values of HGA is different from those of SGA1 at the p-value of 0.1953\% and is different from those of SGA2 at the p-value of 0.3906\%, which suggests the proposed algorithm is statistically better than both SGAs.

\subsection{Comparison with State-of-the-art Multi-agent Organization Design Methodologies}
\subsubsection{Comparison with ODML}
In ODML \cite{Horling-Lesser:2008}, the exploitation of hard constraints, equivalence classes, parallel search, and model abstraction, are used to assist the search process. Rather than going through a decision tree to verify the constraint requirements as ODML does, our algorithm incorporates the array representation that already ensures the satisfaction of constraints in maximum height of the structure and the number of databases in the system. Parallel search and model abstraction are also intuitively used in HGA.

In ODML, the agents are treated in three equivalence classes: the mediators, the aggregators, and the databases. The number of organization alternatives is cut down by discarding organizations which are equivalent to an existing one under the symmetry principle. For instance, the organizations that are symmetrical to each other are equivalent in ODML, and only one should be kept as a candidate.

For the 10 test cases of the IR system, despite the truncation of redundant equivalent organizations, the total number of evaluations needed for ODML can be approximated as $O(2.1^N)$, where N is the number of leaf nodes. In particular, the number of evaluations needed for the 12-database case is 4,304, and that of the 30-database case is 3,788,734,984. Compared with ODML, HGA uses much fewer evaluations. This saves a great amount of computation burden, as the calculation of utility functions can be very computationally expensive.

It should be noted that the HGA is compatible with all the above mentioned search space reducing measures. However, we maintain the equivalent organizations, for they may contribute to finding a good solution. This compromise results in a larger search space for HGA, whereas in ODML, the elimination of redundant equivalent organizations helps to narrow down the search range to a great extent. When the number of equivalent organizations is prevailing, ODML should have an advantage benefited from the elimination measure. Nevertheless, in the studied system, HGA still manages to evolve the population of organizations at a reasonable pace, and it spares the computation time for branch pruning at the same time.

\subsubsection{Comparison with KB-ORG}
Different from ODML, KB-ORG \cite{Sims-et-al:2008} emphasizes the use of design knowledge in application and coordination levels. With good knowledge, a system can be designed with relatively affordable cost. However, in certain cases, design knowledge is hard to acquire. It largely depends on the level of expertise of the designer. Design knowledge is not guaranteed to be accurate, and it needs to be updated following the change of environmental variables.

In the IR model, the main difficulty lies in the coordination of agents, e.g. how many levels of hierarchy is needed. Assume that the designer has successfully located the best organizations for 12, 14, 16, and 18 databases. He may think that a 3-level hierarchy is best for the 20-databse case as well. This will reduce the search space to 58,327 organizations, but it will miss out the highest rated organization, which is 4-leveled with the utility of 821.60. The utility of the best 3-level organization is 814.11, which is worse than the worst utility (820.01) found by HGA within 50,000 evaluations in all runs. On the other hand, if the designer reaches at a relaxed bound of structure height of either 3 or 4 for the 20-database case, the number of organization evaluations will mount to 2,120,662. Although design knowledge could bring convenience, it is sometimes far from satisfactory. In contrast, our algorithm searches for the highest rated organization in a heuristic way. It is able to handle these test cases without the assistance of external expertise.

\section{Conclusions and Future Work}
We have proposed a novel genetic algorithm based approach to solve the problem of designing the best organization in hierarchical multi-agent systems. Complementary to existing methodologies that emphasize on the pruning of the search space, our algorithm uses a bio-inspired evolutionary approach to lead the search to promising areas of the search space, and is thus suitable for optimizing multi-agent systems with a great variety of possible organizations where designer expertise alone is not enough or hard to acquire. In the example of the information retrieval system, we have empirically proved that the algorithm is able to discover competitive baseline structures in different systems and assemble them to obtain the highest rated structure from a magnitude of up to 109 organization alternatives. Moreover, the new crossover and mutation methods helped HGA enhance the search efficiency greatly, promoting its performance both in accuracy and stability.

With necessary modifications, the algorithm is applicable to other models as well. It can be used to optimize any tree-based hierarchical organizations of multi-agent systems, given that proper fitness values are assigned. Application areas include scenario tree and decision tree optimization. On the other hand, the proposed array representation can also be used for other forms of MAS organizations, such as holarchies. We will also explore the incorporation of fuzzy cognitive \cite{Miao-et-al:2002,Song-et-al:2009}, goal-oriented \cite{Yu-et-al:2007}, and inference based \cite{Miao-et-al:2001} methods into the proposed framework to improve its robustness and connection with real-world application domains.

\section{Acknowledgments}
This research is supported by the  National Research Foundation, Prime Minister's Office, Singapore under its IDM Futures Funding Initiative and administered by the Interactive and Digital Media Programme Office.

%
\bibliographystyle{splncs03}
\bibliography{References}  

\begin{thebibliography}{10}
\providecommand{\url}[1]{\texttt{#1}}
\providecommand{\urlprefix}{URL }

\bibitem{Aranha-Iba:2009}
Aranha, C., Iba, H.: The memetic tree-based genetic algorithm and its
  application to portfolio optimization. Memetic Computing  1(2),  139--151
  (2009)

\bibitem{Back:1996}
B\"{a}ck, T.: Evolutionary Algorithms in Theory and Practice: Evolution
  Strategies, Evolutionary Programming, Genetic Algorithms. Oxford University
  Press, Oxford, UK (1996)

\bibitem{Ferber-et-al:2004}
Ferber, J., Gutknecht, O., Michel, F.: From agents to organizations: An
  organizational view of multi-agent systems. In: Giorgini, P., Muller, J.,
  Odell, J. (eds.) Agent-Oriented Software Engineering IV, Lecture Notes in
  Computer Science, vol. 2935, pp. 214--230. Springer Berlin Heidelberg (2004)

\bibitem{Fernandez-Ossowski:2008}
Fern\'{a}ndez, A., Ossowski, S.: Exploiting organisational information for
  service coordination in multiagent systems. In: Proceedings of the 7th
  International Joint Conference on Autonomous Agents and Multiagent Systems
  (AAMAS'08). pp. 257--264 (2008)

\bibitem{Harik:1995}
Harik, G.R.: Finding multimodal solutions using restricted tournament
  selection. In: Proceedings of the 6th International Conference on Genetic
  Algorithms. pp. 24--31 (1995)

\bibitem{Holland:1992}
Holland, J.H.: Adaptation in Natural and Artificial Systems. MIT Press,
  Cambridge, MA, USA (1992)

\bibitem{Horling-Lesser:2008}
Horling, B., Lesser, V.: Using quantitative models to search for appropriate
  organizational designs. Journal of Autonomous Agents and Multi-Agent Systems
  (JAAMAS)  16(2),  95--149 (2008)

\bibitem{Kirley:2006}
Kirley, M.: Dominance hierarchies and social diversity in multi-agent systems.
  In: Proceedings of the 8th Annual Conference on Genetic and Evolutionary
  Computation (GECCO'06). pp. 159--166 (2006)

\bibitem{Lesser:1998}
Lesser, V.: {Reflections on the Nature of Multi-Agent Coordination and Its
  Implications for an Agent Architecture}. Journal of Autonomous Agents and
  Multi-Agent Systems (JAAMAS)  1(1),  89--111 (1998)

\bibitem{Lesser-Corkill:2014}
Lesser, V., Corkill, D.: Challenges for multi-agent coordination theory based
  on empirical observations. In: Proceedings of the 2014 International
  Conference on Autonomous Agents and Multi-agent Systems (AAMAS'14). pp.
  1157--1160 (2014)

\bibitem{Li-et-al:2009}
Li, B., Yu, H., Shen, Z., Miao, C.: Evolutionary organizational search. In:
  Proceedings of the 8th International Conference on Autonomous Agents and
  Multiagent Systems (AAMAS'09). pp. 1329--1330 (2009)

\bibitem{Mezura-Montes:2007}
Mezura-Montes, E.: Evolutionary computation: A unified approach. Artificial
  Life  13(4),  423--426 (2007)

\bibitem{Miao-et-al:2001}
Miao, C., Goh, A., Miao, Y., Yang, Z.: A dynamic inference model for
  intelligent agents. International Journal of Software Engineering and
  Knowledge Engineering  11(05),  509--528 (2001)

\bibitem{Miao-et-al:2002}
Miao, C., Yang, Q., Fang, H., Goh, A.: Fuzzy cognitive agents for personalized
  recommendation. In: Proceedings of the 3rd International Conference on Web
  Information Systems Engineering (WISE'02). pp. 362--371 (2002)

\bibitem{Nan-et-al:2005}
Nan, G., Li, M., Kou, J.: Multi-level genetic algorithm ({MLGA}) for the
  construction of clock binary tree. In: Proceedings of the 7th Annual
  Conference on Genetic and Evolutionary Computation (GECCO'05). pp. 1441--1445
  (2005)

\bibitem{Okamoto-et-al:2008}
Okamoto, S., Scerri, P., Sycara, K.: The impact of vertical specialization on
  hierarchical multi-agent systems. In: Proceedings of the 23rd National
  Conference on Artificial Intelligence (AAAI-08). pp. 138--143 (2008)

\bibitem{Phelps-et-al:2010}
Phelps, S., Mcburney, P., Parsons, S.: Evolutionary mechanism design: A review.
  Journal of Autonomous Agents and Multi-Agent Systems (JAAMAS)  21(2),
  237--264 (2010)

\bibitem{Sims-et-al:2008}
Sims, M., Corkill, D., Lesser, V.: Automated organization design for
  multi-agent systems. Journal of Autonomous Agents and Multi-Agent Systems
  (JAAMAS)  16(2),  151--185 (2008)

\bibitem{Song-et-al:2009}
Song, H., Miao, C., Shen, Z., Miao, Y., Lee, B.S.: A fuzzy neural network with
  fuzzy impact grades. Neurocomputing  72(13),  3098--3122 (2009)

\bibitem{Vazquez-Salceda:2005}
V\'{a}zquez-Salceda, J., Dignum, V., Dignum, F.: Organizing multiagent systems.
  Journal of Autonomous Agents and Multi-Agent Systems (JAAMAS)  11(3),
  307--360 (2005)

\bibitem{Wooldridge-et-al:2000}
Wooldridge, M., Jennings, N., Kinny, D.: The gaia methodology for
  agent-oriented analysis and design. Journal of Autonomous Agents and
  Multi-Agent Systems (JAAMAS)  3(3),  285--312 (2000)

\bibitem{Yang-Luo:2007}
Yang, J., Luo, Z.: Coalition formation mechanism in multi-agent systems based
  on genetic algorithms. Applied Soft Computing  7(2),  561--568 (2007)

\bibitem{Yu-et-al:2013IEEE}
Yu, H., Shen, Z., Leung, C., Miao, C., Lesser, V.R.: A survey of multi-agent
  trust management systems. IEEE Access  1(1),  35--50 (2013)

\bibitem{Yu-et-al:2007}
Yu, H., Shen, Z., Miao, C.: Intelligent software agent design tool using goal
  net methodology. In: Proceedings of the IEEE/WIC/ACM International Conference
  on Intelligent Agent Technology (IAT'07). pp. 43--46 (2007)

\bibitem{Yu-et-al:2012}
Yu, H., Shen, Z., Miao, C., An, B.: Challenges and opportunities for trust
  management in crowdsourcing. In: Proceedings of the IEEE/WIC/ACM
  International Joint Conferences on Web Intelligence and Intelligent Agent
  Technology (WI/IAT'12). pp. 486--493 (2012)

\bibitem{Zafar-et-al:2008}
Zafar, H., Lesser, V., Corkill, D., Ganesan, D.: Using organization knowledge
  to improve routing performance in wireless multi-agent networks. In:
  Proceedings of the 7th International Joint Conference on Autonomous Agents
  and Multiagent Systems (AAMAS'08). pp. 821--828 (2008)

\bibitem{Zhang-et-al:2013}
Zhang, X., Lesser, V.: Meta-level coordination for solving distributed
  negotiation chains in semi-cooperative multi-agent systems. Group Decision
  and Negotiation  22(4),  681--713 (2013)

\end{thebibliography}
%
%

\end{document}